\documentclass{article}
\usepackage{PRIMEarxiv}
\usepackage[T1]{fontenc}
\usepackage[utf8]{inputenc}
\usepackage{amsmath,amssymb,amsthm}
\usepackage{booktabs,array}
\usepackage[hidelinks]{hyperref}
\usepackage{microtype}
\usepackage{url}
\usepackage{fancyhdr}
\usepackage{graphicx}
\graphicspath{{media/}}

\renewcommand{\keywordname}{\textbf{\textit{Keywords}}}
\hypersetup{
  pdftitle={When Information Is Not Enough: Accuracy-Constrained Thermodynamic Costs of Binary Classification},
  pdfsubject={Finite-sample binary inference and bounded-rate memory writing},
  pdfkeywords={rate-distortion, binary inference, stochastic thermodynamics, optimal transport, finite-time control}
}
\newtheorem{theorem}{Theorem}
\newtheorem{proposition}[theorem]{Proposition}
\newtheorem{corollary}[theorem]{Corollary}
\theoremstyle{definition}

\theoremstyle{remark}
\newtheorem{remark}[theorem]{Remark}
\newcommand{\E}{\mathbb{E}}
\newcommand{\Prob}{\mathbb{P}}

\newcommand{\Info}{\mathrm{info}}
\newcommand{\Task}{\mathrm{task}}
\newcommand{\dd}{\mathrm{d}}
\newcommand{\hbin}{h_2}
\newcommand{\W}{\mathcal{W}_1}

\newcommand{\B}{\mathsf{B}}
\newcommand{\MI}{\mathcal I}
\title{When Information Is Not Enough: Accuracy-Constrained Thermodynamic Costs \\ of Binary Classification}
\author{Xuening Wu\\Fudan University\\\texttt{xuening.wu@gmail.com}}
\date{Working manuscript, version 0.3}
\begin{document}
\maketitle
\thispagestyle{empty}
\begin{abstract}
How much entropy must a physical classifier produce to achieve a prescribed accuracy? Rate--distortion theory specifies the minimum information required, but does that information threshold suffice to determine the physical cost? We show that it does not, even for a binary task and a two-state memory.
For a uniform binary target observed through a finite symmetric experiment, replacing the classification-error constraint with its necessary mutual-information threshold strictly lowers the infimum of entropy production under a common operation time, integrated mobility budget, and sufficiently large finite transition-rate cap. The separation holds whenever the target error lies strictly between the Bayes error of the observations and chance.
Two results establish this physical gap. First, ordering observations by posterior confidence gives exact transport--risk and transport--information frontiers. Observations with the same Bayes error can have different cost frontiers. Second, we construct bounded-rate protocols that realize prescribed encoders with write probabilities below one, starting from exact reset, with an explicit excess cost above the transport bound. An achievable information-constrained cost then falls below a lower bound valid for every task-feasible protocol.
Examples with repeated noisy observations illustrate the separation. The results identify a limitation of information-only benchmarks for physical classification: task accuracy and kinetic constraints must be retained explicitly. The cost analyzed is total entropy production during memory writing, excluding data acquisition, controller operation, and subsequent reset.
\end{abstract}


\section{Introduction}
How much entropy must a physical classifier produce to achieve a prescribed accuracy? A classifier must extract useful information from its observations and instantiate that information in a physical memory. The statistical requirement depends on the prediction loss; the cost of writing the memory also depends on its reset state, transition structure, available time, and allowed controls. Connecting these requirements is necessary for assessing the physical cost of a decision.

Rate--distortion theory gives a natural starting point: achieving an error target requires retaining a minimum amount of information about the target~\cite{PolyanskiyWu2022Notes}. Rate--distortion methods also inform Bayesian learning theory~\cite{HafezKolahi2021MER}. Yet the information threshold is only a necessary condition for an individual representation to support the desired accuracy. The question addressed here is whether this distinction changes the minimum physical cost. If the error constraint is replaced by its information threshold, can the resulting optimization select a cheaper representation that fails the task? We compare these two constraints using information about the target itself, without introducing irrelevant variables.

Classification and structured decision systems motivate this distinction: returning a label or probability does not by itself specify the physical cost of meeting an accuracy target. Jev, described by TypeSafe as a typed probabilistic decision interface~\cite{Almeida2026Jev}, is one example of this application context; it is not modeled or evaluated here.

We prove that the information-only constraint can strictly underestimate the optimal entropy production, even for a uniform binary target and a two-state memory. For any informative finite symmetric observation experiment, the separation holds whenever the error target lies strictly above the data Bayes floor and below chance, provided the common transition-rate cap is sufficiently large and finite. Both optimization problems use the same reset state, operation time, and integrated mobility budget, and both admit finite-cost protocols. Thus the result compares feasible physical implementations under common resource constraints.

The proof resolves two linked problems. First, posterior-score ordering gives exact transport--risk and transport--information frontiers and establishes a strict gap between their optimal transport costs. These frontiers depend on the distribution of posterior confidence: one and two noisy observations can have the same Bayes error but different representation costs. Second, an explicit bounded-rate protocol realizes any encoder with positive transport cost and conditional write probabilities bounded away from one. Its excess entropy production above the transport bound is proportional to the inverse rate cap. An achievable information-constrained upper bound can therefore be placed strictly below the task-constrained lower bound (Theorem~\ref{thm:combined}). This construction is needed because a difference between lower bounds alone would not establish a difference between physical cost infima.

The analysis builds on discrete transport thermodynamics~\cite{VuSaito2023}, information-constrained measurement and feedback optimization~\cite{NagaseSagawa2024,NagaseSagawa2025Feedback}, and binary quantization~\cite{NguyenNguyen2020}. Its focus is the cost consequence of replacing task feasibility with a necessary information threshold under a specified finite-rate control model. The optimized cost is joint entropy production during memory writing; acquisition of observations, controller operation, retention, and subsequent reset are excluded. Section~\ref{sec:related} details the relation to prior work. Figure~\ref{fig:overview} summarizes the model and a concrete cost certificate.
\begin{figure}[htbp]
\centering
\includegraphics[width=\linewidth]{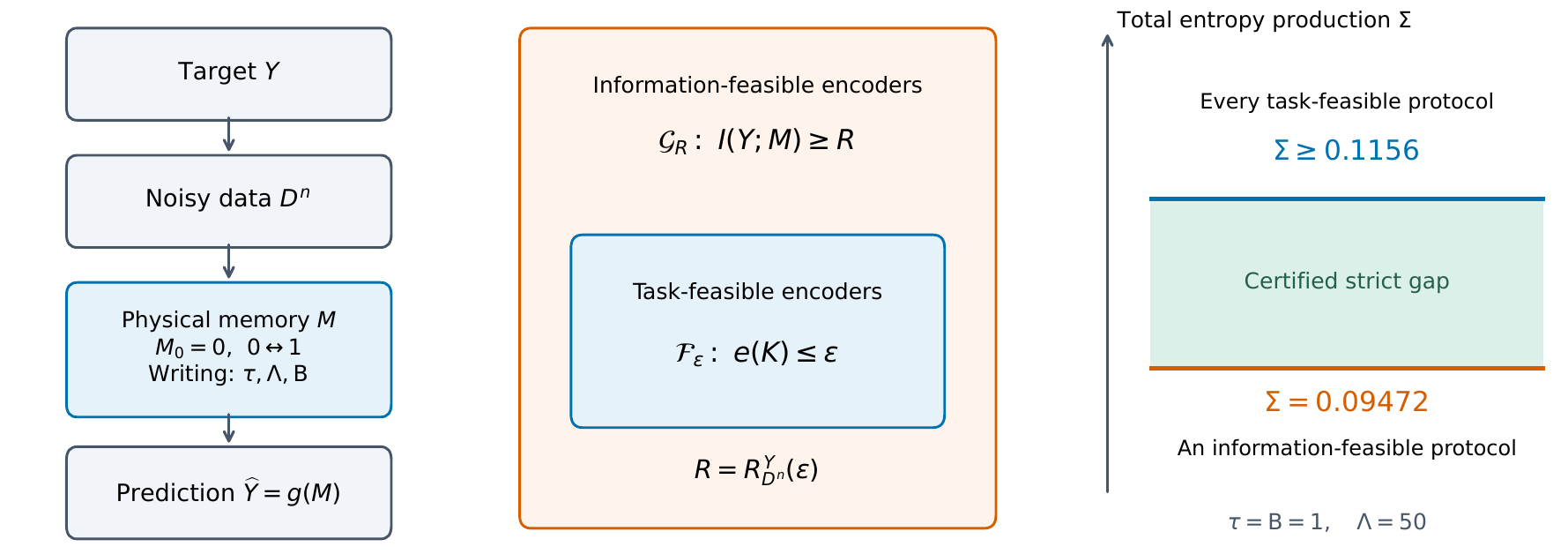}
\caption{Left: noisy observations are encoded into a reset binary memory before decoding. The memory-writing model fixes duration $\tau$, rate cap $\Lambda$, and integrated mobility budget $\B$. Middle: task-feasible encoders form a subset of those satisfying the necessary target-information threshold. This schematic inclusion alone establishes no strict physical-cost gap. Right: a certificate in the stated model, with three conditionally independent binary observations, each having crossover probability $0.2$, at target error $\epsilon=0.2$. At $\tau=\B=1$ and $\Lambda=50$, an explicit information-feasible protocol has total entropy production $0.0947168$, while every task-feasible protocol has cost at least $0.1156$. Task feasibility is established by a separate protocol. These are a feasible upper bound and a universal lower bound, respectively, not two computed physical optima; the displayed separation is specific to the proved conditions.}
\label{fig:overview}
\end{figure}

\section{Related work and attribution}
\label{sec:related}
\paragraph{Rate--distortion and inference.}
The binary rate--distortion formula and Fano argument are standard~\cite{PolyanskiyWu2022Notes}. Hafez-Kolahi et al.~\cite{HafezKolahi2021MER} formulate Bayesian minimum excess risk through constrained rate--distortion optimization. Our threshold $R_D^Y$ optimizes target information over data-mediated binary representations and is explicitly distinguished from $I(D;M)$. The comparison here concerns the physical-cost consequence of relaxing task feasibility to this threshold.

\paragraph{Indirect rate--distortion.}
Noisy source coding studies reconstruction of a latent source when the encoder observes only a corrupted version. Kostina and Verd\'u~\cite{KostinaVerdu2016} derive nonasymptotic bounds and dispersion results for this problem. Under expected distortion, the indirect problem uses the conditional loss $\widetilde\ell(d,\hat y)=\E[\ell(Y,\hat y)\mid D=d]$. Its usual single-letter compression objective is $I(D;\hat Y)$, whereas our relaxation constrains $I(Y;M)$ and our cost is physical transport or entropy production. Our sample size counts observations of one latent bit, rather than a block of source symbols being compressed. Thus the finite-observation calculation should not be interpreted as a new noisy source coding theorem.

\paragraph{Posterior thresholds and quantization.}
Threshold rules based on likelihood ratios are established tools for binary decisions and quantization. Nguyen and Nguyen~\cite{NguyenNguyen2020}, for example, study information-maximizing binary quantizers for noisy observations. Their constrained-output analysis~\cite{NguyenNguyen2020Constrained} is closer to the fixed-memory-mass problem here: it maximizes mutual information for a binary input and continuous observation under constraints on the quantized-output distribution, and establishes a convex-cell structure in posterior coordinates. We use finite observations, allow randomization at the threshold score, and give a direct score-ordering proof for the symmetric case. Our fixed-output-mass frontier follows a fractional-knapsack exchange argument. We claim neither a new threshold principle nor a new general theory of quantization. The symmetric finite-experiment setting permits the exact cost--risk and cost--information comparison used in the physical result.

\paragraph{Transport and physical measurement.}
Vu and Saito~\cite{VuSaito2023} establish the discrete transport relation underlying the entropy-production lower bound. Nagase and Sagawa~\cite{NagaseSagawa2024} optimize information gain under transport constraints and discuss approximately optimal pure-reset protocols. Their binary endpoint specializes to the information-optimal Z-channel in Section~\ref{sec:binary}. Their measurement framework also distinguishes joint entropy production from the memory-side cost containing an information term. Finite-time information-constrained feedback is addressed in Ref.~\cite{NagaseSagawa2025Feedback}.

\paragraph{Kinetic constraints and boundary reachability.}
Diana, Bagci, and Esposito~\cite{DianaBagciEsposito2013} study finite-time erasure with bounded fermionic transition rates. Their Eqs.~(8)--(9) give a minimum transition time and exponential bounds on reachable endpoint probabilities. Our bound $k_d\leq1-e^{-\Lambda\tau}$ uses the same comparison mechanism, applied to data-conditioned writing from reset. Their fixed-relaxation-rate erasure model and heat objective differ from our common rate cap and joint entropy-production objective with a mobility budget. Boundary inaccessibility is therefore an established kinetic effect; the role of our construction is to provide a finite-cap cost certificate for comparing two inference constraints.

The present finite-cap construction provides explicit rates, a common-cap condition, and an exact additive cost for a prescribed encoder with $0\leq k_d\leq\kappa<1$ at a fixed duration and subject to an upper bound on integrated mobility. It is then used to certify strict separation from the task-feasible cost. This comparison, rather than optimal transport or reset regularization in general, is the scope of the physical claim.

\section{Statistical and physical problem setup}
\label{sec:setup}
Let $Y\sim\operatorname{Bernoulli}(1/2)$ be a target and $D$ a finite observation with specified joint law $P_{Y,D}$. An encoder is a stochastic kernel $K(m\mid d)$ for a binary memory $M\in\{0,1\}$. A decoder $g$ predicts $Y$ from $M$. Its optimized Hamming risk is
\begin{equation}
 e(K)=\min_g\Prob(g(M)\neq Y),\qquad
 \mathcal F_\epsilon=\{K:e(K)\leq\epsilon\}.
\end{equation}
Randomized encoders are permitted; a deterministic Bayes decoder suffices. Define the constrained target-information threshold and its relaxation by
\begin{equation}
 R_D^Y(\epsilon)=\inf_{K\in\mathcal F_\epsilon}I(Y;M),\qquad
 \mathcal G_R=\{K:I(Y;M)\geq R\}.
 \label{eq:general-rate}
\end{equation}
For feasible targets, $\mathcal F_\epsilon\subseteq\mathcal G_{R_D^Y(\epsilon)}$. This inclusion gives a weak inequality between minimum costs for any common cost functional. Our question concerns strictness. The quantity $I(Y;M)$ describes information about the inference target; it is not the data-compression quantity $I(D;M)$. All encoders must factor through the observations.

Write $w_d=P(D=d)>0$ and $k_d=K(1\mid d)$. The memory starts exactly at $M_0=0$. Data are frozen during writing, and only unit-length edges $(d,0)\leftrightarrow(d,1)$ are allowed. Initial and final joint laws are
\begin{equation}
 p_0(d,m)=w_d\mathbf1_{\{m=0\}},\qquad p_K(d,m)=w_dK(m\mid d).
\end{equation}
Transport is restricted to each data sector. Moving its terminal mass $w_dk_d$ across one edge gives
\begin{equation}
 \W(p_0,p_K)=\sum_dw_dk_d=\Prob(M=1)=q(K).
 \label{eq:data-transport}
\end{equation}
The physical reset label remains fixed throughout the comparison. Complementing an encoder means preparing a different endpoint from that same reset state.

We first minimize this transport cost over the two endpoint constraint sets. We then minimize the total entropy production of continuous-time writing protocols under common duration $\tau$, integrated mobility budget $\B$, and transition-rate cap $\Lambda$. Section~\ref{sec:control} defines that model precisely. The controller may access observed $D$, but cannot access an unobserved $Y$. The costs of obtaining the data, computing or operating the controller, retaining the memory, decoding, and subsequent reset are outside the writing model.

All logarithms are natural, and $\hbin(x)=-x\ln x-(1-x)\ln(1-x)$, with continuous values at zero and one. Entropy production is measured in units of Boltzmann's constant.

\section{Exact boundaries for finite symmetric experiments}\label{sec:stat}
\subsection{Posterior-confidence geometry}
Let $Y\sim\operatorname{Bernoulli}(1/2)$ and let $D$ have a finite alphabet. Assume the observation experiment is symmetric: there is an involution $T$ on observations such that
\begin{equation}
 P(D=T(d)\mid Y=0)=P(D=d\mid Y=1).
\end{equation}
This includes repeated binary symmetric observations that are independent conditional on $Y$. Define
\begin{equation}
 w_d=P(D=d),\qquad s_d=2P(Y=1\mid D=d)-1.
\end{equation}
Then $w_{T(d)}=w_d$, $s_{T(d)}=-s_d$. Set
\begin{equation}
 S=\sum_d w_d(s_d)_+,\qquad r=\tfrac12-S,
\end{equation}
where $(z)_+=\max\{z,0\}$ and $r$ is the Bayes error using all observations. Assume $S>0$.

For a kernel $k=(k_d)$, write
\begin{equation}
 q=\sum_dw_dk_d,\qquad A=\sum_dw_ds_dk_d.
\end{equation}
The induced source--memory channel obeys
\begin{equation}
 a=P(M=1\mid Y=0)=q-A,\qquad b=P(M=1\mid Y=1)=q+A.
\end{equation}
Consequently its Bayes error and target information are
\begin{equation}
 e(k)=\tfrac12-|A|,\qquad
 I(Y;M)=\MI(q,A):=\hbin(q)-\tfrac12\hbin(q-A)-\tfrac12\hbin(q+A).
 \label{eq:IA}
\end{equation}
The physical controller uses only $D$. The source--memory channel is induced by averaging the encoder over $P(D\mid Y)$; no access to the latent label is assumed.

\subsection{Exact cost--risk and cost--information boundaries}
Sort observations so that $s_1\geq\cdots\geq s_m$, retaining their weights $w_i$. Put $Q_j=\sum_{i\leq j}w_i$, $V_j=\sum_{i\leq j}w_is_i$, with $Q_0=V_0=0$. For $Q_{j-1}\leq q\leq Q_j$, define
\begin{equation}
 A_*(q)=V_{j-1}+s_j(q-Q_{j-1}).
 \label{eq:confidence-frontier}
\end{equation}
We use $0\leq q\leq1/2$. The associated kernel selects observations in descending order of posterior score, with possible randomization at one boundary score. Ties can be handled in any order.

\begin{theorem}[Exact finite-observation boundaries]
\label{thm:frontier}
For every $q\in[0,1/2]$,
\begin{align}
 \max_{k:\ \sum w_dk_d=q}\left|\sum_dw_ds_dk_d\right|&=A_*(q),\\
 \min_{k:\ \sum w_dk_d=q}e(k)&=\tfrac12-A_*(q),\\
 \max_{k:\ \sum w_dk_d=q}I(Y;M)&=\Psi(q):=\MI(q,A_*(q)).
 \label{eq:psi}
\end{align}
For $r\leq\epsilon<1/2$, let $\delta=1/2-\epsilon$. Then the exact task transport is
\begin{equation}
 q_T(\epsilon)=\min\{q\in[0,1/2]: A_*(q)\geq\delta\}.
 \label{eq:qT}
\end{equation}
The minimum target information among task-feasible kernels equals
\begin{equation}
 R_D^Y(\epsilon)=\ln2-\hbin(\epsilon).
 \label{eq:RD}
\end{equation}
For the information-only relaxation at this same threshold,
\begin{equation}
 q_I(\epsilon)=\min\{q\in[0,1/2]:\Psi(q)\geq R_D^Y(\epsilon)\}.
 \label{eq:qI}
\end{equation}
These minima are attained. The quantity $R_D^Y$ concerns $I(Y;M)$, not data compression $I(D;M)$.
\end{theorem}
\begin{proof}
For a fixed $q$, maximize $A=\sum_iw_is_ik_i$ under $0\leq k_i\leq1$ and $\sum_iw_ik_i=q$. If $s_i>s_j$ while $k_i<1$ and $k_j>0$, transferring an equal amount of selected probability from $j$ to $i$ strictly increases $A$ and preserves $q$. Repeating such exchanges gives the threshold kernel and \eqref{eq:confidence-frontier}. Symmetry yields the negative counterpart with the same $q$, so the maximum absolute value is also $A_*(q)$. The error formula in \eqref{eq:IA} proves the second claim.

For $A>0$ in the interior of the channel domain,
\begin{equation}
 \frac{\partial\MI(q,A)}{\partial A}
 =\tfrac12\ln\frac{(q+A)(1-q+A)}{(q-A)(1-q-A)}>0.
\end{equation}
The function is even in $A$, with continuous boundary values. Thus maximizing absolute contrast also maximizes information at fixed $q$. This proves \eqref{eq:psi}. Encoders of cost $q>1/2$ can be complemented to encoders of cost $1-q<1/2$, preserving both Bayes error and mutual information. Hence no optimum requires $q>1/2$, and \eqref{eq:qT} follows.

For \eqref{eq:RD}, the binary Fano inequality gives $I(Y;M)\geq\ln2-\hbin(\epsilon)$ for any task-feasible encoder. A symmetric Bayes rule, with a fair coin on posterior ties, induces a binary symmetric channel with crossover probability $r$. Flipping its decision independently with probability $(\epsilon-r)/(1-2r)$ produces a binary symmetric channel of crossover probability $\epsilon$, which achieves equality. This encoder factors through $D$, so the lower bound is achievable within the observation constraint. The fixed-cost information boundary now gives \eqref{eq:qI}. Compactness of the finite-dimensional kernel cube and continuity establish attainment.
\end{proof}

\begin{remark}[No monotonicity shortcut]
$A_*(q)$ is piecewise affine and nondecreasing on $[0,1/2]$, but $\Psi(q)$ need not be globally increasing. For example, observations with posterior score zero can leave $A_*$ constant while increasing $q$ reduces information. Thus \eqref{eq:qI} is a first-crossing definition, not an assumed global inverse. On each score segment the channel varies affinely, and mutual information is convex in that channel. The minimum feasible $q$ can be obtained by scanning segment endpoints and solving the first threshold crossing in a qualifying segment.
\end{remark}

\begin{theorem}[Strict finite-observation representation gap]
\label{thm:stat-gap}
For every informative finite symmetric experiment and every
\begin{equation}
 r<\epsilon<1/2,
\end{equation}
one has $q_I(\epsilon)<q_T(\epsilon)$. In particular, the gap survives arbitrary nonzero observation noise in this class, provided the requested error lies strictly above the data Bayes floor.
\end{theorem}
\begin{proof}
Let $p_+=P(s_D>0)\leq1/2$. Since $0<\delta<S$, the threshold encoder reaches contrast $\delta$ before exhausting all positive scores. Therefore $0<q_T<p_+\leq1/2$, and its Bayes error is exactly $\epsilon$.

We show that Fano is strict for this encoder. If equality held in $H(Y\mid M)\leq\hbin(\epsilon)$, the binary error indicator $E$ for a Bayes decoder would be independent of $M$. Since the error is less than $1/2$, the decoder is nonconstant. Uniformity of $Y$ would then give
\begin{equation}
 \tfrac12=\epsilon+(1-2\epsilon)P(g(M)=1),
\end{equation}
forcing $P(g(M)=1)=1/2$, hence $q_T=1/2$, a contradiction. Thus $\Psi(q_T)>R_D^Y(\epsilon)$. Continuity, $\Psi(0)=0<R_D^Y(\epsilon)$, and the intermediate value theorem yield a feasible information cost strictly below $q_T$.
\end{proof}

\subsection{What happens at the boundaries}
At $\epsilon=1/2$, both optimal transport costs are zero. At $\epsilon<r$, the statistical task is infeasible. At $\epsilon=r$, the task cost is $q_T=p_+$, where $p_+=\Prob(s_D>0)$. A strict gap occurs exactly when some $q\in[0,p_+)$ satisfies $\Psi(q)\geq\ln2-\hbin(r)$. If $p_+<1/2$, then $r>0$ and the strict-Fano proof already gives a strict gap at the floor. If $p_+=1/2$, equality or strictness must be determined from the information boundary; the interior theorem does not decide this case.

A single binary symmetric observation with noise $\nu\in(0,1/2)$ gives $r=\nu$, $A_*(q)=(1-2\nu)q$ and
\begin{equation}
 \Psi_\nu(q)=\hbin(q)-\tfrac12\hbin(2\nu q)
                         -\tfrac12\hbin(2(1-\nu)q).
\end{equation}
This function is strictly increasing: the channel at smaller $q$ is obtained from the larger-$q$ channel by independently discarding some $M=1$ events, with strict loss of information for a nontrivial source correlation. Hence $q_I=q_T=1/2$ at $\epsilon=\nu$. The gap can therefore vanish at the statistical floor, despite being positive at every interior target.

\section{Direct observation as a solvable special case}\label{sec:binary}
For direct observation $D=Y$, write
\begin{equation}
 a=\Prob(M=1\mid Y=0),\qquad b=\Prob(M=1\mid Y=1),\qquad q=(a+b)/2.
 \label{eq:ab}
\end{equation}
The joint graph and reset convention are those of Section~\ref{sec:setup}. Equation~\eqref{eq:data-transport} specializes to
\begin{equation}
 \W(p_0,p_K)=q,\label{eq:transport}
\end{equation}
and the error and information are
\begin{align}
 e(K)&=(1-|b-a|)/2,\label{eq:bayes}\\
 I(K)&=\hbin(q)-\tfrac12\hbin(a)-\tfrac12\hbin(b).\label{eq:mi}
\end{align}
Set
\begin{equation}
 R(\epsilon)=\ln2-\hbin(\epsilon),\qquad 0\leq\epsilon\leq1/2.\label{eq:binary-rd}
\end{equation}
For $0\leq\epsilon\leq1/2$ and $0\leq R\leq\ln2$, the two endpoint optimizations reduce to
\begin{align}
 W^*_{\Task}(\epsilon)&=\min_{a,b\in[0,1]:e(K)\leq\epsilon}q,\\
 W^*_{\Info}(R)&=\min_{a,b\in[0,1]:I(K)\geq R}q.\label{eq:two-optima}
\end{align}
The following specialization of Theorem~\ref{thm:frontier} has a self-contained proof in Appendix~\ref{app:binary-proof}.

\subsection{Exact direct-observation gap}
\begin{theorem}[Exact binary representation gap]
\label{thm:binary}
Fix $0<\epsilon<1/2$ and define
\begin{equation}
 \phi(q)=\hbin(q)-\tfrac12\hbin(2q),\qquad 0\leq q\leq\tfrac12.
 \label{eq:phi}
\end{equation}
There is a unique $q_\epsilon\in(0,1/2)$ satisfying
\begin{equation}
 \phi(q_\epsilon)=R(\epsilon)=\ln2-\hbin(\epsilon).
 \label{eq:q-root}
\end{equation}
The exact optimal costs are
\begin{equation}
 W^*_{\Task}(\epsilon)=\tfrac12-\epsilon,
 \qquad W^*_{\Info}(R(\epsilon))=q_\epsilon,
\end{equation}
and their difference is strictly positive:
\begin{equation}
 \Delta_W(\epsilon)=\tfrac12-\epsilon-q_\epsilon>0.
 \label{eq:strict-gap}
\end{equation}
A task-optimal encoder is $(a,b)=(0,1-2\epsilon)$. An information-only optimal encoder is $(a,b)=(0,2q_\epsilon)$, which satisfies
\begin{equation}
 I(K)=R(\epsilon),\qquad e(K)=\tfrac12-q_\epsilon>\epsilon.
\end{equation}
Swapping the source labels gives equivalent optimizers without changing the physical reset state.
\end{theorem}

\subsection{Three distinct optimization objectives}
The information--transport boundary $\phi$ and its Z-channel optimizer are specializations of the measurement result in Ref.~\cite{NagaseSagawa2024}, Eqs.~(9) and (14), under a uniform binary source. The task-constrained comparison is the focus here.

Table~\ref{tab:channels} compares three encoders. Minimizing information at a prescribed error, minimizing transport at that error, and minimizing transport at a prescribed information threshold select different channels.
\begin{table}[htbp]
\centering
\small
\setlength{\tabcolsep}{5pt}
\begin{tabular}{@{}lcccc@{}}
\toprule
Encoder & $(a,b)$ & Bayes error & $I(Y;M)$ & $\W$\\
\midrule
Rate--distortion & $(\epsilon,1-\epsilon)$ & $\epsilon$ & $R(\epsilon)$ & $1/2$\\
Task transport & $(0,1-2\epsilon)$ & $\epsilon$ & $\phi(1/2-\epsilon)$ & $1/2-\epsilon$\\
Information transport & $(0,2q_\epsilon)$ & $1/2-q_\epsilon$ & $R(\epsilon)$ & $q_\epsilon$\\
\bottomrule
\end{tabular}
\caption{Exact channel comparison for $0<\epsilon<1/2$. The last row violates the target error despite attaining its necessary information threshold.}
\label{tab:channels}
\end{table}
The first and third rows carry identical information but have different risks and transport costs. The second row is task-feasible at lower transport than the first, but carries strictly more information. The result concerns the feasible-set relaxation, not a contradiction of rate--distortion coding achievability.

\section{Repeated noisy observations}
\subsection{Binomial experiment}
Fix $0<\nu<1/2$. Let $D^n=(X_1,\ldots,X_n)$ with $X_i=Y\oplus Z_i$, where $\oplus$ denotes addition modulo two and the $Z_i$ are independent Bernoulli$(\nu)$ noise variables, independent of $Y$. The count $j=\sum_iX_i$ is sufficient for both the objective and constraints: any kernel on the full observation sequence can be averaged over strings with the same count, preserving its cost and induced source--memory channel. Define
\begin{align}
 p_1(j)&=\binom{n}{j}(1-\nu)^j\nu^{n-j},\\
 p_0(j)&=\binom{n}{j}\nu^j(1-\nu)^{n-j},\\
 w_j&=\tfrac12[p_1(j)+p_0(j)],\qquad
 s_j=\frac{p_1(j)-p_0(j)}{p_1(j)+p_0(j)}.
\end{align}
Here $p_y(j)=\Prob(\sum_i X_i=j\mid Y=y)$. For $0<\nu<1/2$, $s_j$ increases with $j$, so an optimal encoder uses an upper threshold on the count with possible randomization at the threshold count. The Bayes error is
\begin{equation}
 r_n=\Prob\{\operatorname{Bin}(n,\nu)>n/2\}
       +\tfrac12\Prob\{\operatorname{Bin}(n,\nu)=n/2\}.
\end{equation}
Equations~\eqref{eq:confidence-frontier}--\eqref{eq:qI} determine both optima from at most $n+1$ score levels.

\subsection{Same Bayes error, different transport costs}
One and two noisy observations have the same Bayes error $r_1=r_2=\nu$. For two observations, the tied count carries posterior score zero. At the Bayes floor $\epsilon=\nu$,
\begin{equation}
 q_{T,1}=\tfrac12,\qquad
 q_{T,2}=\frac{(1-\nu)^2+\nu^2}{2}<\tfrac12.
\end{equation}
The one-observation gap is zero, whereas the two-observation gap is strictly positive. For $\nu=0.2$, the costs are
\begin{center}
\begin{tabular}{@{}cccc@{}}\toprule
Samples & Bayes error & $q_T$ & $q_I$\\\midrule
1 & $0.2$ & $0.5$ & $0.5$\\
2 & $0.2$ & $0.34$ & $0.300667502$\\\bottomrule
\end{tabular}
\end{center}
The same distinction also holds away from the floor: at $\epsilon=0.25$ and $\nu=0.2$, one observation gives $(q_T,q_I)=(0.416666667,0.408012516)$, whereas two give $(0.283333333,0.224286655)$. Both have Bayes error $0.2$. This is not reproduced merely by substituting the Bayes error into the original noiseless formula: identical Bayes errors and identical task-information thresholds have different cost frontiers. The distribution of posterior confidence matters.

\subsection{Sample dependence and limiting behavior}
At any fixed feasible error target, both $q_{T,n}$ and $q_{I,n}$ are nonincreasing in $n$, because one can ignore an additional observation. Their difference need not be monotone. For $\nu=0.2$ and $\epsilon=0.2$, numerical evaluation gives
\begin{center}
\begin{tabular}{@{}cccc@{}}\toprule
$n$ & $q_{T,n}$ & $q_{I,n}$ & Difference\\\midrule
2 & $0.340000000$ & $0.300667502$ & $0.039332498$\\
3 & $0.340000000$ & $0.249472892$ & $0.090527108$\\
4 & $0.314400000$ & $0.246584757$ & $0.067815243$\\
5 & $0.304647619$ & $0.236834298$ & $0.067813321$\\
10 & $0.300058947$ & $0.228641285$ & $0.071417662$\\\bottomrule
\end{tabular}
\end{center}
For fixed $\nu\in(0,1/2)$ and $\epsilon\in(0,1/2)$, the target is feasible for all sufficiently large $n$, and the two optima converge to their noiseless counterparts as $n\to\infty$. Indeed, $r_n\to0$, the noiseless optima are lower bounds because data-mediated channels form a subset of all source channels, and the symmetric majority decision induces a binary symmetric channel of error $r_n$. Writing $\delta=1/2-\epsilon$, thinning its positive output yields feasible encoders with task cost $\delta/(1-2r_n)$ and information boundary $\Psi_{r_n}(q)$. These converge to the noiseless task cost $\delta$ and information boundary $\hbin(q)-\tfrac12\hbin(2q)$, respectively. Continuity and strict increase of the latter bound its threshold crossing and complete the squeeze argument.

\section{Physical writing under a finite rate cap}\label{sec:control}
\subsection{Admissible physical model}
Fix a duration $\tau>0$, an integrated mobility budget $\B>0$, and a common rate cap $\Lambda>0$. On every data sector, allow measurable, piecewise continuous rates
\begin{equation}
 0\leq r_d^+(t),r_d^-(t)\leq\Lambda.
\end{equation}
An edge may be switched off in both directions. On active edges the construction below has both rates positive and admits local detailed balance with time-dependent energies. There is no additional lower rate bound or independent energy/barrier bound. For a bath at temperature $T>0$, our local detailed balance convention on an active edge is
\begin{equation}
 E_{d,1}(t)-E_{d,0}(t)=-k_{\rm B}T\ln\frac{r_d^+(t)}{r_d^-(t)}.
\end{equation}
The admissible control class permits the data-dependent energy difference and kinetic prefactor to be selected independently, subject to the rate cap. Equivalently, any positive rate pair can be represented with prefactor $a_d(t)=\sqrt{r_d^+(t)r_d^-(t)}$ and the energy difference above. The constructed rates give finite energy differences on active sectors; inactive sectors have both rates zero. This is a specified control model, not a demonstration that a fixed device with bounded energy range or restricted barrier control can implement every protocol. Exact reset is an initial resource whose preparation cost is excluded.

Let $x_d(t)=\Prob(M_t=1\mid D=d)$. The joint forward and backward probability fluxes are $u_d=w_d(1-x_d)r_d^+$ and $v_d=w_dx_dr_d^-$. Define
\begin{equation}
 J_d=u_d-v_d,\quad F_d=\ln(u_d/v_d),\quad
 \sigma(t)=\sum_d J_dF_d,\quad
 m(t)=\sum_d\frac{J_d}{F_d}.
\end{equation}
For $u_d=v_d>0$, the mobility contribution is defined by continuity as $J_d/F_d=u_d$, and the entropy-production contribution is zero. If both fluxes vanish, both contributions are zero. If exactly one flux vanishes, the mobility contribution is zero and the entropy-production density has its extended value $+\infty$; an isolated divergence at the reset boundary is treated through the improper time integral. Let $\Sigma=\int_0^\tau\sigma\,\dd t$ and $\mathcal M=\int_0^\tau m\,\dd t$. We impose $\mathcal M\leq\B$.

An admissible protocol $\Pi$ determines a terminal encoder $K_\Pi$. Define
\begin{align}
 C_{T,\Lambda}(\epsilon)&=\inf_{\Pi}\{\Sigma[\Pi]:e(K_\Pi)\leq\epsilon,\ \mathcal M[\Pi]\leq\B\},\\
 C_{I,\Lambda}(R)&=\inf_{\Pi}\{\Sigma[\Pi]:I(Y;M_\tau)\geq R,\ \mathcal M[\Pi]\leq\B\},
\end{align}
where both infima use the same reset state, graph, duration $\tau$, and rate cap $\Lambda$. Infeasible problems have value $+\infty$. We also use $C_\Lambda(K)$ for the infimum with a prescribed endpoint encoder. The infima over protocols, including those with prescribed endpoints, are not assumed to be attained. Mobility is the logarithmic-mean kinetic quantity above, not expected jump activity.

For any admissible writing protocol, the standard discrete transport bound gives
\begin{equation}
 \Sigma\mathcal M\geq q^2,\qquad \Sigma\geq q^2/\B.
 \label{eq:lower}
\end{equation}
Here it also follows directly from $\sum_d|J_d|\leq\sqrt{m\sigma}$ and integrated mass conservation. This is an application of existing transport thermodynamics, not a new inequality~\cite{VuSaito2023,NagaseSagawa2024}.

\begin{theorem}[Explicit protocol under a uniform finite rate cap]
\label{thm:regularization}
Fix an encoder with $0\leq k_d\leq\kappa<1$ and $q>0$. Put
\begin{equation}
 F=\frac q\B,\qquad c=\frac1{e^F-1}.
\end{equation}
Suppose
\begin{equation}
 \Lambda>\frac c\tau,
 \qquad
 \Lambda\geq\frac{\kappa(1+c)}{\tau(1-\kappa)}.
 \label{eq:rate-threshold}
\end{equation}
There is a protocol starting exactly from reset, reaching the prescribed encoder exactly at $\tau$, with all rates bounded by $\Lambda$, for which
\begin{equation}
 \mathcal M\leq\B,\qquad
 \Sigma=\frac{q^2}{\B}+\frac q{\tau\Lambda}\ln(1+c).
 \label{eq:regularized-cost}
\end{equation}
Thus no singular transition rate is required. If $C_\Lambda(K)$ is the protocol cost infimum for this endpoint, then
\begin{equation}
 0\leq C_\Lambda(K)-q^2/\B
 \leq\frac q{\tau\Lambda}\ln(1+c)
\end{equation}
whenever \eqref{eq:rate-threshold} holds.
\end{theorem}
\begin{proof}
For $k_d=0$, keep both rates zero. For $k_d>0$, prescribe the conditional memory probability
\begin{equation}
 x_d(t)=\Prob(M_t=1\mid D=d)=\frac{k_dt}{\tau}.
\end{equation}
Define $\rho(t)=\min\{\Lambda t,c\}$ and rates
\begin{equation}
 r_d^-(t)=\begin{cases}\Lambda,&t=0,\\ \rho(t)/t,&t>0,\end{cases}
 \qquad
 r_d^+(t)=\frac{k_d}{\tau}\frac{1+\rho(t)}{1-k_dt/\tau}.
 \label{eq:explicit-rates}
\end{equation}
The two rates are finite and continuous across $t_0=c/\Lambda$. Direct substitution gives
\begin{equation}
 (1-x_d)r_d^+-x_dr_d^-=k_d/\tau,
\end{equation}
so the master equation is satisfied exactly from $x_d(0)=0$. The backward rate is bounded by $\Lambda$. The forward rate obeys
\begin{equation}
 r_d^+(t)\leq\frac{k_d(1+c)}{\tau(1-k_d)}
 \leq\frac{\kappa(1+c)}{\tau(1-\kappa)}\leq\Lambda.
\end{equation}
The log ratio of the forward and backward rates is finite on every active sector, including $t=0$; hence the rates can be assigned a finite instantaneous energy difference under local detailed balance.

Writing $J_d=w_dk_d/\tau$, the joint fluxes are
\begin{equation}
 u_d=J_d(1+\rho(t)),\qquad v_d=J_d\rho(t).
\end{equation}
Every active sector has the same flux force
\begin{equation}
 f(t)=\ln(1+1/\rho(t))\geq\ln(1+1/c)=F\quad(t>0).
\end{equation}
Consequently $\mathcal M=(q/\tau)\int_0^\tau f(t)^{-1}\,\dd t\leq q/F=\B$. The force diverges logarithmically as $t\downarrow0$, due to the pure initial distribution, but the rates do not diverge and the entropy integral is finite. Since $t_0<\tau$,
\begin{align}
 \Sigma
 &=\frac q\tau\left[\int_0^{c/\Lambda}\ln\left(1+\frac1{\Lambda t}\right)\dd t
                  +(\tau-c/\Lambda)F\right]\\
 &=qF+\frac q{\tau\Lambda}\ln(1+c).
\end{align}
The last equality follows from
$\int_0^c\ln(1+1/z)\dd z=(c+1)\ln(c+1)-c\ln c$.
Combining this explicit upper bound with \eqref{eq:lower} proves the claim.
\end{proof}

\begin{remark}
The excess-cost bound concerns a family of finite-cap models as $\Lambda$ grows. It does not assert that the transport lower bound is attained at any fixed cap. It also does not extend to a prescribed positive lower rate bound on every edge: inactive sectors are part of the present admissible model.
\end{remark}

\subsection{Direct-observation specialization}
For direct observation $D=Y$ and $0<\epsilon<1/2$, write $q_T=1/2-\epsilon$, and let $q_I$ be the root of
\begin{equation}
 \hbin(q_I)-\tfrac12\hbin(2q_I)=\ln2-\hbin(\epsilon).
\end{equation}
Theorem~\ref{thm:binary} gives $0<q_I<q_T<1/2$. The information-optimal endpoint is $k_0=0,k_1=2q_I$, so all conditional write probabilities are bounded away from one.

\begin{corollary}[Finite-cap strict physical gap]
\label{cor:direct}
Let $C_{T,\Lambda}$ and $C_{I,\Lambda}$ be the actual protocol infima under the task and information constraints, respectively, with the same $\tau,\B,\Lambda$. Put $c_I=(e^{q_I/\B}-1)^{-1}$. If the rate conditions of Theorem~\ref{thm:regularization} hold for both $(0,2q_I)$ and $(0,2q_T)$, and
\begin{equation}
 \Lambda>\frac{\B q_I\ln(1+c_I)}{\tau(q_T^2-q_I^2)},
 \label{eq:strict-cap}
\end{equation}
then both task and information problems are feasible and
\begin{equation}
 C_{T,\Lambda}-C_{I,\Lambda}
 \geq\frac{q_T^2-q_I^2}{\B}
       -\frac{q_I\ln(1+c_I)}{\tau\Lambda}>0.
 \label{eq:cap-gap}
\end{equation}
All rate-cap thresholds in this corollary are finite for $0<\epsilon<1/2$.
\end{corollary}
\begin{proof}
Every task endpoint has $q\geq q_T$, giving $C_{T,\Lambda}\geq q_T^2/\B$. Theorem~\ref{thm:regularization} supplies an actual information-feasible protocol with cost at most the right-hand side of \eqref{eq:regularized-cost} at $q_I$, as well as a finite-cost task-feasible protocol at $q_T$. Subtract the information upper bound from the task lower bound. This compares achievable costs with a universal lower bound, not two lower bounds.
\end{proof}
For $\epsilon=0.25$, $\tau=\B=1$, and $\Lambda=30$, one obtains
\begin{equation}
 q_I\simeq0.165035659,\quad
 C_{I,\Lambda}\leq0.037595380,
 \quad C_{T,\Lambda}\geq0.0625.
\end{equation}
The certified gap in total entropy production is at least $0.024904620$ in units of $k_{\rm B}$. The information witness uses integrated mobility of approximately $0.913748547$, within the common unit budget.

\section{A strict physical gap with noisy observations}\label{sec:combined}
The statistical optimum may assign $k_d=1$ on some observations. Such endpoints cannot be reached exactly from reset with a finite rate cap in finite time. It would therefore be invalid simply to substitute both statistical optima into an assumed physical equality. The next theorem avoids that error by using an information-feasible witness whose write probabilities are bounded away from one.

\begin{theorem}[Strict physical gap with noisy observations and finite rates]
\label{thm:combined}
Fix an informative finite symmetric experiment, $r<\epsilon<1/2$, $\tau>0$, and $\B>0$. Use the physical model of Section~\ref{sec:control}, with rates controlled only by observed $D$. There exists a finite $\Lambda_0$ such that for every $\Lambda>\Lambda_0$, both optimization problems are feasible and
\begin{equation}
 C_{T,\Lambda}(\epsilon)>C_{I,\Lambda}(R_D^Y(\epsilon)).
\end{equation}
More explicitly, choose any encoder with cost $0<q<q_T$, information at least $R_D^Y(\epsilon)$, and $\max_d k_d=\kappa<1$. Put $c=(e^{q/\B}-1)^{-1}$. Whenever its rate conditions hold and a task-feasible encoder with all write probabilities below one is also realizable under the same cap,
\begin{equation}
 C_{T,\Lambda}-C_{I,\Lambda}
 \geq\frac{q_T^2-q^2}{\B}
       -\frac{q\ln(1+c)}{\tau\Lambda}.
 \label{eq:combined-certificate}
\end{equation}
The right-hand side is positive for a sufficiently large finite cap.
\end{theorem}
\begin{proof}
Theorem~\ref{thm:stat-gap} proves that the task-optimal threshold encoder $k^T$ has information strictly above $R_D^Y(\epsilon)$. Replace it by $k_d=\alpha k_d^T$ for $0<\alpha<1$ sufficiently close to one. Continuity preserves information strictly above the threshold, while $q=\alpha q_T<q_T$ and $\max k_d\leq\alpha<1$. Thus the information witness required by the theorem always exists.

To ensure task feasibility with finite rates, put $\delta=1/2-\epsilon$ and let $k^B$ be a Bayes rule with contrast $S$, randomizing ties symmetrically. Choose $\beta\in(\delta/S,1)$ and set $k_d=\beta k_d^B$. This has contrast $\beta S>\delta$, hence error strictly below $\epsilon$, and all encoder probabilities are below one. Theorem~\ref{thm:regularization} realizes both witnesses at sufficiently large finite caps.

Every task-feasible protocol has $q\geq q_T$, and hence entropy production at least $q_T^2/\B$. The explicit information witness has entropy production given by \eqref{eq:regularized-cost}. Subtracting gives \eqref{eq:combined-certificate}. All rate and strictness thresholds are finite since both witnesses have maximum write probability below one and $q<q_T$. This also establishes a common $\Lambda_0$ valid for every larger cap.
\end{proof}

\subsection{A concrete finite-cap, finite-sample certificate}
Take $n=3$, $\nu=0.2$, and $\epsilon=0.2$. The count statistic, mixture probabilities, posterior scores, and information-feasible encoder are
\begin{center}
\begin{tabular}{@{}ccccc@{}}\toprule
Count $j$ & 3 & 2 & 1 & 0\\\midrule
$w_j$ & $0.26$ & $0.24$ & $0.24$ & $0.26$\\
$s_j$ & $63/65$ & $0.6$ & $-0.6$ & $-63/65$\\
Information witness $k_j$ & $0.92$ & $0.23$ & $0$ & $0$\\\bottomrule
\end{tabular}
\end{center}
Here $r_3=0.104$, $q_T=0.34$, and the information witness has
\begin{align}
 q&=0.2944,& A&=0.26496,\\
 e&=0.23504,& I(Y;M)&\simeq0.196607776>R(0.2)\simeq0.192744757.
\end{align}
It fails the task, as intended for a relaxation witness, but satisfies the correct target-information threshold. For $\tau=\B=1$, $\Lambda=50$, the rate conditions hold: the forward-rate bound is approximately $45.094227<50$. The explicit protocol gives the following bounds:
\begin{equation}
 C_{I,50}\leq0.094716766,
 \qquad C_{T,50}\geq0.34^2=0.1156.
\end{equation}
A task-feasible encoder with maximum write probability below one is $k_j=29/33$ on $j=2,3$ and zero otherwise. It has error $0.152<0.2$, and its rate conditions also hold at $\Lambda=50$. Therefore both problems have finite feasible costs and
\begin{equation}
 \boxed{C_{T,50}-C_{I,50}\geq0.020883234>0.}
\end{equation}
The certificate does not require knowing either finite-cap optimum exactly.

\subsection{A statistical--kinetic feasibility obstruction}
\begin{proposition}[Bayes-floor obstruction at finite rate]
For every protocol with forward rate bounded by $\Lambda$ and duration $\tau$,
\begin{equation}
 e\geq r+e^{-\Lambda\tau}(1/2-r).
 \label{eq:feasibility}
\end{equation}
This lower bound remains valid with an additional mobility budget. In particular, an informative experiment cannot attain its exact data Bayes error in this reset-memory model at any finite $\Lambda\tau$.
\end{proposition}
\begin{proof}
The conditional probability satisfies $\dot x_d\leq\Lambda(1-x_d)$ and $x_d(0)=0$, so $k_d\leq1-e^{-\Lambda\tau}=:\gamma$. By symmetry,
\begin{equation}
 |A|=\left|\sum_dw_ds_dk_d\right|\leq\gamma\sum_dw_d(s_d)_+=\gamma S.
\end{equation}
Thus $e=1/2-|A|\geq1/2-\gamma S=r+e^{-\Lambda\tau}(1/2-r)$.
\end{proof}
No attainability assertion is made for \eqref{eq:feasibility} under finite entropy production or a mobility budget. It is a necessary condition. A necessary duration-rate condition for target $\epsilon>r$ is
\begin{equation}
 \Lambda\tau\geq\ln\frac{1/2-r}{\epsilon-r}.
\end{equation}
This explains why the combined theorem requires an interior error target and why the statistical floor examples must not be interpreted as finite-rate physical implementations.

\section{Numerical verification}
The supplementary reproduction program \texttt{verify\_extensions.py} checks the score-sorting frontier against independent linear programs for $n=1,\ldots,15$, noise levels $0.05,0.2,0.4,0.49$, and 19 cost values per experiment (1,140 checks). It also compares random encoders with the information frontier and checks the strict gap at interior targets. Direct quadrature verifies the regularization cost and mobility, and the master equation is integrated using the explicit rates to check the endpoint probabilities. An additional 200 linear-program checks cover 40 randomly generated finite symmetric experiments, including experiments with posterior ties. Numerical integration of the master equation checks both three-sample witness protocols against a maximum sampled path-error tolerance of $10^{-8}$; the information witness uses integrated mobility $0.974250814$. The observed path discrepancy was below $2\times10^{-15}$ for these prescribed linear trajectories in the recorded run; this is not a general solver accuracy guarantee. Numerical results support the calculations; the theorems rest on the analytic proofs above.

\paragraph{Reproducibility.}
The supplied program \texttt{verify\_extensions.py} generates the numerical results and the two verification figures; \texttt{create\_overview.py} generates the introduction schematic. NumPy generators use seeds 729 and 219 for the binomial and general symmetric checks, respectively. Linear programs use SciPy's HiGHS backend, with objective agreement thresholds $10^{-7}$ and $10^{-9}$ for the two suites. Information threshold roots use Brent's method. Master equations use RK45 with relative tolerance $10^{-10}$, absolute tolerance $10^{-12}$, and maximum step $0.002$ in the unit-duration examples; path discrepancies are evaluated at 1,001 equally spaced times. Entropy-production and mobility integrals use adaptive quadrature with requested absolute tolerance $10^{-12}$. The result file \texttt{extension\_verification.json} records the software versions, seeds, and solver settings. All observations are specified analytically or generated synthetically; no empirical dataset is used. The accompanying source package contains the programs, numerical output, and instructions for reproduction.

% Figure~\ref{fig:frontiers} visualizes the statistical frontiers and the finite-cap certificate. Figure~\ref{fig:protocols} shows the numerical checks of the explicit physical witnesses. These are numerical illustrations and verification of the stated model, not measurements from an experimental device.

\begin{figure}[p]
\centering
\includegraphics[width=\linewidth]{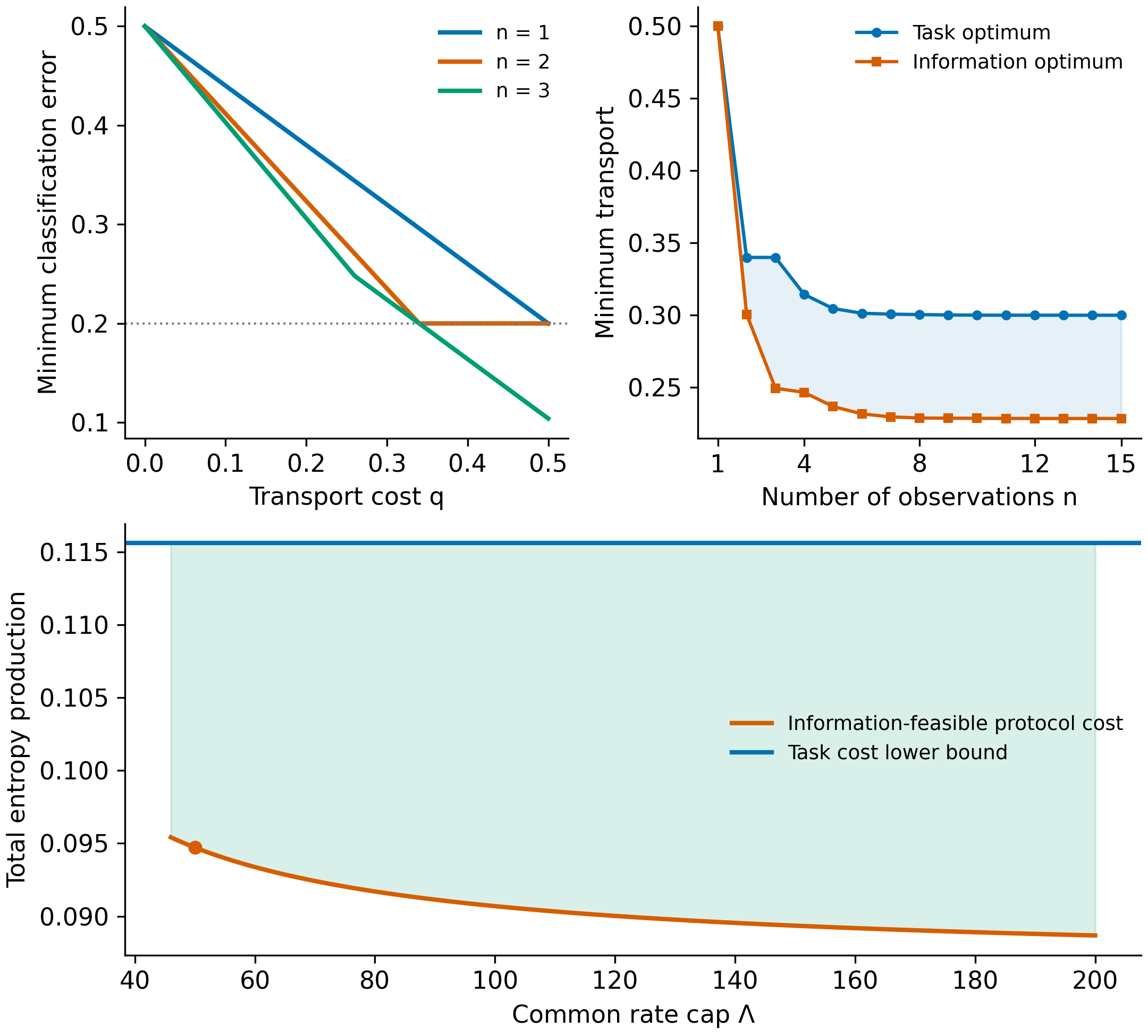}
\caption{Conditionally independent binary observations with crossover probability $0.2$. Top left: the exact minimum classification error at fixed transport cost for one, two, and three observations; the dotted line marks error $0.2$. Top right: task and information transport optima at target error $0.2$ as sample size varies. For one and two observations, the target equals the statistical Bayes floor; the task-optimal endpoints are therefore not attainable by finite-rate protocols. This obstruction does not automatically apply to the information-only endpoints. Bottom: for three observations and $\tau=\B=1$, the entropy production of the information-feasible encoder $(k_3,k_2,k_1,k_0)=(0.92,0.23,0,0)$ under the explicit protocol, versus the universal task-cost lower bound $0.1156$. The orange curve is an achievable upper bound, not an asserted optimum. Shading marks the certified separation, and the dot marks $\Lambda=50$.}
\label{fig:frontiers}
\end{figure}

\begin{figure}[p]
\centering
\includegraphics[width=\linewidth]{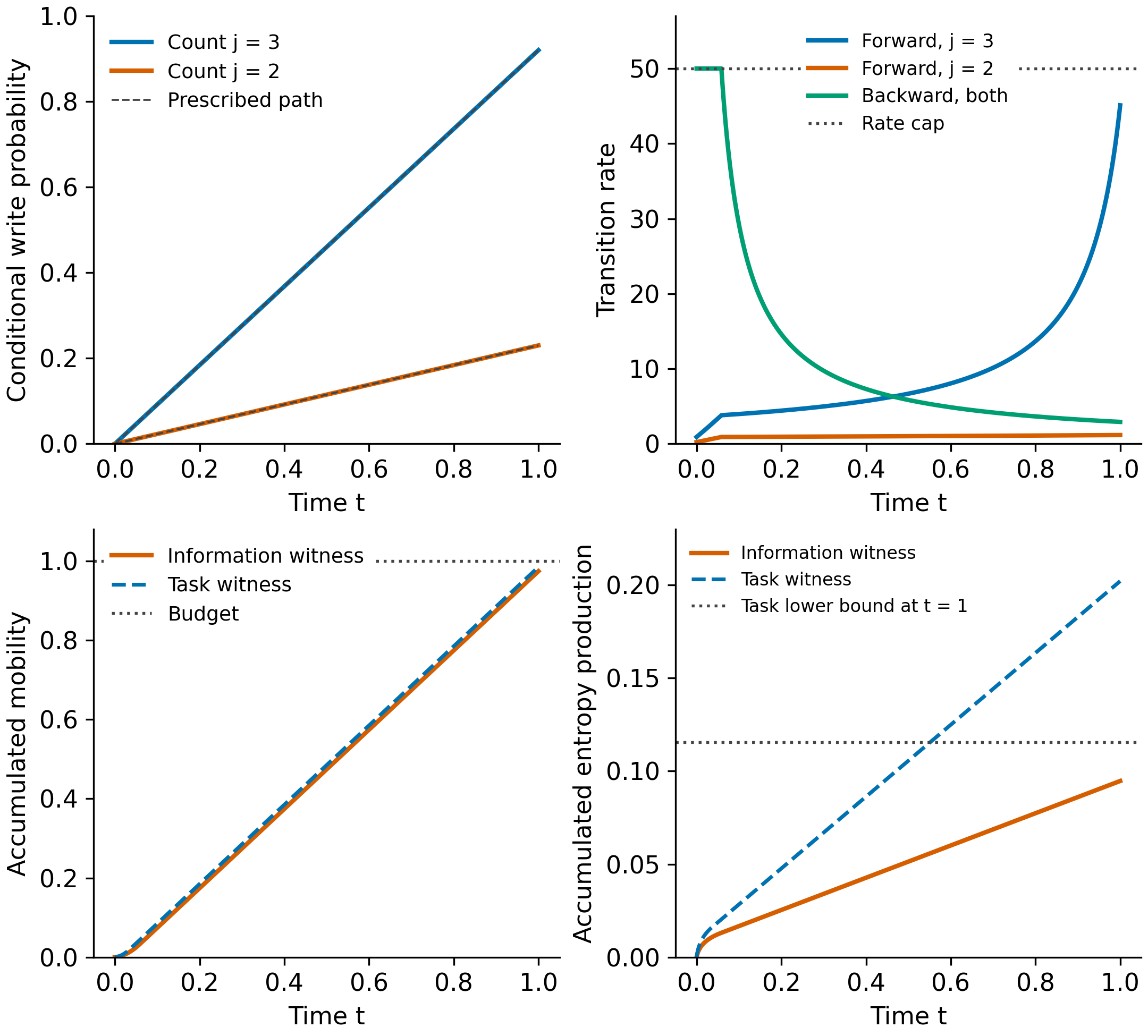}
\caption{Verification of the three-observation protocols with crossover probability $0.2$, $\tau=\B=1$, and common rate cap $\Lambda=50$. Top left: numerically integrated conditional write probabilities for the information witness, with prescribed linear paths overlaid. Top right: its forward and backward transition rates; inactive sectors remain off. Bottom left: accumulated dynamical mobility for the information witness and the task-feasible witness $(k_3,k_2,k_1,k_0)=(29/33,29/33,0,0)$, both within budget. Bottom right: accumulated total entropy production. The horizontal reference is the task lower bound at the terminal time $t=1$, not a bound at intermediate times. The task witness demonstrates feasibility and is not claimed to be optimal.}
\label{fig:protocols}
\end{figure}

% Flush the two verification figures before starting the discussion.
\clearpage
\section{Discussion}
\subsection{Why the gap occurs}
The information threshold is a necessary condition for task performance. At the task-optimal transport endpoint, the binary Fano inequality is strict whenever the output distribution is unbalanced and the error lies in $(0,1/2)$. This creates room to reduce transport while retaining the threshold amount of information. The finite-observation theorem identifies the relevant unbalanced endpoint through posterior-score ordering. Its proof uses familiar information inequalities and a threshold-selection argument; its content is the exact comparison of the two optimization problems.

Posterior confidence matters beyond the Bayes error. The one-versus-two-sample example holds the Bayes error fixed while changing the cost frontier, and the three-sample certificate shows that this statistical distinction survives a common finite kinetic constraint. The underlying task remains binary inference rather than general supervised prediction with arbitrary test inputs and hypothesis classes.

\subsection{What finite-cap realizability adds}
A strict difference between transport lower bounds does not itself establish a difference between optimal physical costs. Theorem~\ref{thm:regularization} supplies a finite-rate upper bound for an actual information-feasible protocol. Comparing it with the task lower bound proves Theorem~\ref{thm:combined}. The physical conclusion no longer assumes equality in an ideal-control variational relation at a pure reset boundary.

Some statistically optimal encoders require conditional write probability one and cannot be reached in finite time under a finite rate cap. The proof therefore constructs witnesses bounded away from unit write probability rather than treating all statistical endpoints as physically attainable. The feasibility obstruction further shows that the exact data Bayes floor is inaccessible at finite $\Lambda\tau$ in this reset-memory model. At a fixed integrated mobility budget, the limiting term $q^2/\B$ has no separate duration dependence; a convention $\B=\tau\bar m$ instead produces a $1/\tau$ term. Rate feasibility must be checked under either convention.

\subsection{Entropy production, heat, and work}
The optimized quantity is total entropy production of the data--memory system, not heat alone. For one bath at temperature $T$ and heat $Q_{\rm bath}$ positive into the bath,
\begin{equation}
 \Sigma=\Delta H(D,M)+\frac{Q_{\rm bath}}{k_{\rm B}T},\qquad
 \Delta H(D,M)=H(M\mid D)=\sum_dw_d\hbin(k_d).
\end{equation}
If the latent variable $Y$ is included in the frozen state, conditional independence $M\perp Y\mid D$ gives the same entropy change and entropy production. Endpoint entropy varies with the encoder, so a total-entropy-production gap need not equal a heat gap. Work also depends on endpoint energies and the adopted free-energy accounting.

The memory-side quantity $\Delta H(M)+Q_{\rm bath}/(k_{\rm B}T)$ equals total entropy production plus $I(D;M)$ in the present data-conditioned writing model. For direct observation $D=Y$, this becomes the target information, but the two mutual informations differ for noisy data. We therefore do not transfer our cost-optimality statements to memory-side dissipation without reoptimizing that different objective.

\subsection{Implications for classification and structured decision systems}
The result distinguishes an information requirement from a decision requirement: satisfying a necessary target-information threshold does not ensure the prescribed classification accuracy, and optimizing under the relaxed constraint can strictly underestimate the required writing cost. For systems that return labels or probabilities, cost--accuracy comparisons must therefore specify the decision loss and retain task feasibility explicitly.

Structured interfaces such as Jev~\cite{Almeida2026Jev} motivate this question but are outside the model. Type validity, probability calibration, and classification accuracy are distinct properties; the present theorems concern only binary error and the entropy production of writing a two-state memory. Applying the analysis to a deployed classifier would require its physical transition structure, admissible controls, and a cost boundary that accounts for computation, data movement, and repeated operation. The results establish neither a calibration guarantee nor an energy advantage for any existing AI system.

\section{Limitations}
The statistical theorem assumes a uniform binary target, a finite symmetric observation experiment, a binary memory, and Hamming loss. It does not establish the same strictness for arbitrary priors, losses, observation models, or memory graphs. Finite-sample results concern repeated noisy observations of a binary latent target, not general learning sample complexity or generalization error. The boundary at the data Bayes floor must be treated separately, and the gap need not grow monotonically with sample size.

The physical model permits arbitrary data-conditioned, time-dependent rates below a common cap and allows edges to be switched off in both directions. It does not impose a common positive lower rate, fixed hardware energy range, restricted control access, or a cost for the controller. Pure reset is supplied as an initial resource. Although our protocols have finite rates and finite energy differences on each active edge, no independent energy or barrier budget is enforced across models.

The finite-cap statements provide constructive upper bounds and strict-gap certificates, not exact finite-cap cost optima. The cap may need to be large, especially near feasibility boundaries. The model excludes data acquisition and processing costs outside memory writing, as well as retention and subsequent erasure. Thus the results do not determine a complete inference-cycle heat or work budget.

Numerical checks support the formulas and constructed dynamics in the specified examples; they are not experimental evidence. The comparisons with prior work delimit the contribution to the stated inference constraints and control model; they do not establish a general optimal-control principle.

\section{Conclusion}
For finite symmetric binary experiments, posterior-score ordering yields exact transport--risk and transport--information boundaries. Every error target strictly above the data Bayes floor and below chance exhibits a strict gap between task-constrained transport and its target-information relaxation. An explicit finite-rate protocol turns this statistical gap into a strict difference between actual entropy-production infima under a common duration, mobility budget, and sufficiently large finite rate cap. The resulting limitation of information-only cost prediction survives finite noisy observations and does not require singular control rates. Its scope is the writing stage of binary inference, where task feasibility, posterior confidence, and kinetic constraints must be retained together.

\section*{Data and code availability}
The study uses analytically specified and synthetic observations; no proprietary or empirical dataset is required. The accompanying source package includes the verification and figure-generation programs, numerical results, and dependency versions. The author makes these materials available for public distribution with the manuscript.

\section*{Ethics statement}
This theoretical and computational study involved no human participants, animal experiments, or personal data.

\appendix
\section{Complete proof of the direct-observation binary theorem}
\label{app:binary-proof}

\begin{proof}
\textit{Step 1: the information threshold is necessary and achievable.}
For a Bayes decoder $\widehat Y=g(M)$, define the error indicator $E=\mathbf 1_{\{Y\neq\widehat Y\}}$. Since $Y$ is binary, $Y$ is determined by $(M,E)$, and $E$ is determined by $(M,Y)$. Consequently
\begin{equation}
 H(Y\mid M)=H(E\mid M)\leq H(E)=\hbin(e(K)).
\end{equation}
As $e(K)\leq1/2$ and $\hbin$ increases on $[0,1/2]$, task feasibility implies $I(Y;M)\geq\ln2-\hbin(\epsilon)$. Conversely, the binary symmetric channel $(a,b)=(\epsilon,1-\epsilon)$ has Bayes error $\epsilon$ and mutual information $\ln2-\hbin(\epsilon)$. Thus \eqref{eq:binary-rd} is exactly the minimum information among task-feasible binary encoders. This is the binary rate--distortion calculation; see also~\cite{PolyanskiyWu2022Notes}.

\textit{Step 2: solve the task-constrained transport problem.}
Equation~\eqref{eq:bayes} gives $|b-a|\geq1-2\epsilon$. Since $a,b\geq0$,
\begin{equation}
 q=\frac{a+b}{2}\geq\frac{|b-a|}{2}\geq\tfrac12-\epsilon.
\end{equation}
The encoder $(a,b)=(0,1-2\epsilon)$ attains both inequalities and has error exactly $\epsilon$. Therefore $W^*_{\Task}(\epsilon)=1/2-\epsilon$.

\textit{Step 3: maximize information at fixed transport.}
For $q\in[0,1/2]$, the constraint $a+b=2q$ gives $a\in[0,2q]$ and $b=2q-a$. The function
\begin{equation}
 f_q(a)=\hbin(a)+\hbin(2q-a)
\end{equation}
is concave on that interval; hence its minimum is achieved at an endpoint. Both endpoints have value $\hbin(2q)$. Substituting into \eqref{eq:mi} proves
\begin{equation}
 I(K)\leq\phi(q),
\end{equation}
with equality at $(a,b)=(0,2q)$ or $(2q,0)$. Moreover,
\begin{equation}
 \phi(0)=0,\quad \phi(1/2)=\ln2,\quad
 \phi'(q)=\ln\frac{2(1-q)}{1-2q}>0
 \quad(0<q<1/2).
 \label{eq:phi-derivative}
\end{equation}
Continuity and strict monotonicity establish existence and uniqueness of $q_\epsilon$.

If an encoder has $q<q_\epsilon$, then $q<1/2$ and $I(K)\leq\phi(q)<R(\epsilon)$, so it is infeasible for the information benchmark. Encoders with $q>1/2$ cannot improve on the candidate $q_\epsilon<1/2$. The endpoint encoder $(0,2q_\epsilon)$ attains the threshold and proves $W^*_{\Info}(R(\epsilon))=q_\epsilon$. This also justifies the global optimization over all $a,b\in[0,1]$, rather than only those with $q\leq1/2$.

\textit{Step 4: establish strictness.}
Define the continuous function on $[0,1/2]$
\begin{equation}
 G(\epsilon)=\phi(1/2-\epsilon)-\ln2+\hbin(\epsilon).
\end{equation}
Its endpoint values are $G(0)=G(1/2)=0$, while for $0<\epsilon<1/2$,
\begin{align}
 G'(\epsilon)
 &=-\ln\frac{1+2\epsilon}{2\epsilon}
   +\ln\frac{1-\epsilon}{\epsilon}\\
 &=\ln\frac{2(1-\epsilon)}{1+2\epsilon}.
\end{align}
This derivative is positive on $(0,1/4)$, zero at $1/4$, and negative on $(1/4,1/2)$. Hence $G(\epsilon)>0$ throughout the open interval. Therefore
\begin{equation}
 \phi(1/2-\epsilon)>R(\epsilon)=\phi(q_\epsilon).
\end{equation}
Strict monotonicity of $\phi$ yields $q_\epsilon<1/2-\epsilon$. Finally, \eqref{eq:bayes} gives $e(0,2q_\epsilon)=1/2-q_\epsilon>\epsilon$.
\end{proof}
\end{document}